\documentclass[sigconf]{acmart}
\AtBeginDocument{%
  }

\setcopyright{acmlicensed}
\copyrightyear{2025}
\acmYear{2025}
\acmDOI{}
\acmISBN{}

\usepackage{xcolor}
\definecolor{okabe1}{HTML}{000000}
\definecolor{okabe2}{HTML}{E69F00}
\definecolor{okabe3}{HTML}{56B4E9}
\definecolor{okabe4}{HTML}{009E73}
\definecolor{okabe5}{HTML}{F0E442}
\definecolor{okabe6}{HTML}{0072B2}
\definecolor{okabe7}{HTML}{D55E00}
\definecolor{okabe8}{HTML}{CC79A7}

\usepackage{hyperref}

\usepackage[]{todonotes} 

\begin{document}

\title{Privacy-Preserving Learning-Augmented Data Structures}


\author{Prabhav Goyal}
\email{prabhavg@uci.edu}
\affiliation{%
  \institution{University of California, Irvine}
  \city{Irvine}
  \country{USA}
}

\author{Vinesh Sridhar}
\email{vineshs1@uci.edu}
\orcid{0009-0009-3549-9589}
\affiliation{%
  \institution{University of California, Irvine}
  \city{Irvine}
  \country{USA}
}

\author{Wilson Zheng}
\email{wilsonz5@uci.edu}
\affiliation{%
  \institution{University of California, Irvine}
  \city{Irvine}
  \country{USA}
}


\begin{abstract}
\emph{Learning-augmented data structures} use predicted frequency estimates to retrieve frequently occurring database elements faster than standard data structures. 
Recent work has developed data structures that optimally exploit these frequency estimates while maintaining robustness to adversarial prediction errors. 
However, the privacy and security implications of this setting remain largely unexplored. 

In the event of a security breach, data structures should reveal minimal information beyond their current contents. 
This is even more crucial for learning-augmented data structures, whose layout adapts to the data. 
A data structure is \emph{history independent} if its memory representation reveals no information about past operations except what is inferred from its current 
contents. 
In this work, we take the first step towards privacy and security guarantees in this setting by proposing the first learning-augmented data structure that is strongly history independent, robust, and supports dynamic updates. 


To achieve this, we introduce two techniques: \emph{thresholding}, which automatically makes any learning-augmented data structure robust, and \emph{pairing}, a simple technique that provides strong history independence in the dynamic setting.
Our experimental results demonstrate a tradeoff between security and efficiency but are still competitive with the state of the art.
\end{abstract}



\maketitle

\section{Introduction}

Dictionaries are a fundamental class of data structures that support 
updates and retrievals of key-value pairs. Many classic 
data structures, such as binary search trees and hash tables, 
efficiently implement dictionaries.
In general, data structures 
that support \emph{inexact searches}, such as predecessor and range 
queries, achieve retrieval times of at best $O(\log n)$. However, 
these structures often ignore information about the keys' access frequencies. 

Recently, Lin {\it et al}.~\cite{lin2022learning} introduced the first 
\emph{learning-augmented search tree}. Here, each key $k_i$ is assigned 
an access frequency estimate $f_i \in (0,1]$  by a machine learning model. 
The authors design a search tree that can retrieve any key $k_i$ in 
$O(\log 1/f_i)$ time while supporting inexact queries~\cite{lin2022learning}. 
This property is known as \emph{consistency} w.r.t the assigned 
frequency distribution. 
Zeynali {\it et al}.~\cite{zeynali2024robust} subsequently proposed a 
learning-augmented dictionary that 
also achieves \emph{robustness} against 
adversarial models; that is, it guarantees that no key ever exceeds a 
depth of $O(\log n)$ even if an adversarial model can arbitrarily edit frequency estimates. Other recent works that study consistent and robust learning-augmented data structures include~\cite{chen2025on,fu2024learning}. 
See also, e.g.,  ~\cite{kraska2018case, benomar2024learning,ferragina2020learned,sato2023fast,ding2020alex,nathan2020learning,chen2015hiflash,polak2024learning,qin2025piecewise,mitzenmacher2018model,mitzenmacher2019scheduling,purohit2018improving,benomar2025tradeoffs,hsu2019learning,mitzenmacher2022algorithms}, which apply machine learning predictions to improve classical algorithms in general. 

In contrast to numerous works on performance guarantees, few works have studied privacy 
and security needs in this setting. 
Since these structures adapt to 
their input distributions, they leak more information about the underlying data than 
their classic counterparts, making them more vulnerable to data breaches. 
Even the splay tree~\cite{sleator1985self}, another structure which 
achieves robustness and consistency among other nice properties, 
fails to prevent such leakage as its layout prioritizes recent insertions 
and queries. 

This issue is amplified in today's computing landscape, where large-scale 
data breaches are increasingly common. In 2024 alone, more than 3,000
confirmed breaches in the United States affected an estimated 1.35 billion people, more than 91\% of which were attributed to cyberattacks~\cite{data-breach}. 
These risks are projected to grow as artificial intelligence continues to 
integrate deeply into critical infrastructure. 

One way to mitigate this growing problem is by implementing \emph{history independence}~\cite{hartline2005characterizing,naor2001anti}. 
A data structure is history independent if its internal memory 
representation depends solely on its current key-value contents and not 
on the sequence of past operations. This restricts the leakage 
of historical information,  providing protection against data 
breaches. 


In this work, we take the first step towards privacy protection in the learning-augmented setting. 
We present the first robust, consistent, dynamic, and history independent learning-augmented data structure.
We give results for both weak and strong history independence which presents a tradeoff between security and efficiency. 
To achieve these guarantees, we introduce two general techniques. 
First, \emph{thresholding}, which makes any consistent learning-augmented data structure robust at little cost.
Second, \emph{pairing}, a new construction that maintains two coordinated copies of the dataset to ensure strong history independence in the dynamic setting. 

Our experimental results show that thresholding produces learned data structures that match the performance of the state-of-the-art while offering improvements in space usage, 
and that pairing gives strongly history independent data structures at the cost of a $2\times$ increase in search times.
This establishes thresholding and pairing as simple mechanisms to construct learning-augmented data structures that are both efficient and secure. 

\section{History Independent Learned Data Structures}


For a review of history independence, see Appendix~\ref{sec:hist-ind-ex}. 
We begin by showing that the only robust learning-augmented data structure 
in the literature that supports dynamic updates, 
\texttt{RobustSL}~\cite{zeynali2024robust}, is \emph{not} history independent. 
Specifically, we exploit its amortized update scheme by constructing a 
counterexample where two distinct operation sequences produce the same 
external state but yield distinguishable internal representations. 

We first review \texttt{RobustSL}'s update scheme. Let $n$ denote the current number 
of elements and $N$ denote a cutoff value used in the update scheme.\footnote{Zeynali {\it et al.}~\cite{zeynali2024robust} use $n_t$ and $n$ in place of $n$ and $N$.}
$N$ 
is initialized to 4 and $n$ begins at 0. Elements with frequencies less 
than $1/N$ are handled by \texttt{RobustSL} in a special way. To achieve robustness, $O(\log N)$ must equal $O(\log n)$ at all times, so the structure is
periodically destroyed and rebuilt to update $N$. 


If, after an insertion, $n$ equals $N$, then we set $N \leftarrow N^2$ and reconstruct the structure using this new cutoff value in $O(n\log n)$ time. If, after a deletion, $n$ equals $N^{1/4}$, then we set $N \leftarrow \sqrt N$ and do the same. An amortized analysis shows that this produces $O(\log n)$-time insertion and deletion with high probability in $n$~\cite{zeynali2024robust}.\footnote{
An event succeeds with high probability (w.h.p.) in $n$ if it succeeds with probability greater than $1 - 1/n$. 
} 

Now consider any arbitrary state, $A$, of the algorithm and two distinct sets 
of operations, $X$ and $Y$, which nevertheless take $A$ to 
the same final state, $B$. Let $c \leftarrow (N-n)$. In $X$, we insert $c$ identical 
elements into the structure and then remove 1 of those elements. In $Y$, we 
insert $c-1$ identical elements into the structure. Both sets of 
operations take us from $A$ to $B$, however $X$ makes $N$ change, whereas 
$Y$ does not. Since $N$ influences the location of elements in the structure, 
one can distinguish the memory representations created by $X$ and $Y$ despite 
the fact that both add the same content to $A$. Thus, \texttt{RobustSL}, 
under the given updating scheme, is not history independent.

\subsection{A Weakly History Independent Update Scheme}\label{sec:history-ind-scheme}
Now we present the first 
history independent, robust, and dynamic learning-augmented data structure.
We do so by modifying the update scheme of \texttt{RobustSL} such that it is 
now history independent (in the static setting, skip-lists such as \texttt{RobustSL}, have been shown to be history independent~\cite{golovin2008uniquely}.
Thus, it is sufficient to provide a history independent
update scheme to obtain history independence in the dynamic setting). 

Our scheme is based off of a scheme by Hartline {\it et al.}~\cite{hartline2005characterizing} that constructs a weakly history independent 
dynamic hash table.
We appear to be the first that extends it to a data structure with logarithmic-time updates as opposed 
to constant-time updates. 
Their scheme has two basic ideas. First, rather than having an upper and lower cutoff for changing $N$,
each change in $N$ is associated with a single value of $n$ (e.g., every time $n$ doubles, double $N$ and rebuild).
This way, every value of $n$ is associated with a canonical $N$, avoiding the issue we raise in the
previous section. 

However, this cutoff scheme can be abused by an adversary to trigger repeated rebuilds by inserting and deleting elements around a given cutoff value. 
Thus, their second contribution introduces randomness in the location of these cutoffs.
Specifically, $N$ is now conceived as a random variable that has a $O(1/n)$ chance 
of updating after each operation.
As in, e.g., ~\cite{hartline2005characterizing,goodrich2017auditable,molnar2006tamper,chen2015hiflash,goodrich2016more}, we assume that this source of randomness is kept private 
from an adversary, which prevents the above issue. 

Their scheme works like so. Before each insertion, if $N = n$, then randomly choose a value in $\{n+1, \ldots, 2(n+1) - 1\}$ and rebuild the structure with $N$ set to that value. Otherwise $N > n$, and they do the following. With probability $1/(n+1)$, rebuild the structure with $N \leftarrow 2n$ and with probability $1/(n+1)$ do the same with $N \leftarrow 2n+1$. Otherwise, do not rebuild. After this the new item is inserted. Deletions work similarly. We initially delete the element. Then, if $n \leq N/2$, resize $N$ to a value uniformly chosen between $\{n, \ldots 2n-1\}$. If $n > N/2$, then set $N \leftarrow n$ with probability $1/n$. Otherwise, do nothing. 

Hartline {\it et al.} show that each update has an $O(1/n)$ chance of rebuilding~\cite{hartline2005characterizing}. 
The cost of rebuilding a robust learning-augmented data structure is $O(n\log n)$, since the structure has depth $O(\log n)$. Thus, in our setting, the expected update cost is $O((n\log n) / n) = O(\log n)$. We conclude the following.

\begin{theorem}\label{thm:hist-ind}
There exists a weakly history independent, consistent, and robust learning-augmented data structure that supports $O(\log n)$-time dynamic updates in expectation.
\end{theorem}
\begin{proof}
Follows from the above updating scheme applied to \texttt{RobustSL}~\cite{zeynali2024robust}. At all times in the above scheme, $N$ differs from $n$ by at most a constant factor, so $O(\log N) \subseteq O(\log n)$. 
Therefore, \texttt{RobustSL} still achieves robustness under this scheme. 
\end{proof}

Next, we develop several more robust, weakly history independent, and dynamic data structures using our new thresholding scheme.

\section{Thresholding}\label{sec:thresholding}

We now present our second contribution, \emph{thresholding}. 
This technique is a simple modification of the
learning-augmented framework that allows
any consistent learning augmented data structure $\mathcal D$ 
to become robust with a negligible effect
on search times. 
We contrast our general method with prior approaches that focus on making individual data structures robust (cf.~\cite{zeynali2024robust,fu2024learning}).

The two prior works on robust learning-augmented data structures both group together all keys with a frequency below a threshold (e.g., $1/n$) into a non-learned structure that has a standard $O(\log n)$ retrieval time~\cite{zeynali2024robust,fu2024learning}. Our new insight is to instead modify the frequencies themselves according to thresholds, allowing the data structure to act as a black box. 
It has been conjectured that methods similar to ours may work in practice (see OpenReviews of~\cite{fu2024learning,chen2025on}). 
We are the first to show that this is true theoretically.
For now, we assume a static setting, where the total number of keys, $n$, is known in advance. We later lift this assumption. 

\begin{definition}[Threshold Frequency Scheme]
In the learning-augmented setting, each key $k_i$ has a corresponding frequency $f_i$. 
A threshold frequency scheme defines a new frequency 
$$f_i' \leftarrow \max\{f_i/2, 1/(2n)\}.$$ 
Each key is then inserted with frequency $f_i'$ rather than $f_i$.
\end{definition}

\begin{theorem}\label{thm:thresh-learning}
Any consistent learning-augmented data structure $\mathcal D$ can be made robust via our threshold frequency scheme.
\end{theorem}
\begin{proof}
By definition, for any assignment of frequencies that sum to no more than 1, $\mathcal D$ guarantees that each $k_i$ is retrieved in time at most $O(\log 1/f_i)$. We first confirm that our new thresholding scheme has frequencies that sum to no more than 1. Indeed, since at most $n$ keys can increase in frequency by at most $1/(2n)$, we have that $\sum_{i=1}^n f_i' \leq \sum_{i=1}^nf_i/2 + n \times 1/(2n) \leq 1/2 + 1/2 = 1$.

Under this scheme, each $k_i$ is retrieved in time at most $O(\log 2/f_i) \subseteq O(\log 1/f_i)$, so consistency w.r.t. the original frequency assignment is preserved. All $f_i'\geq 1/(2n)$, so any key's depth is at worst $O(\log n)$. 
Thus, $\mathcal D$ is now robust. 
\end{proof}

Theorem~\ref{thm:thresh-learning} and prior work implies new consistent and robust implementations of treaps~\cite{lin2022learning,chen2025on} and B-treaps~\cite{chen2025on}. We emphasize that any new work in this field now only needs to satisfy consistency as our thresholding scheme can immediately be applied to achieve robustness. 

Next, we extend our result to another field of work called \emph{biased data structures} and show that any biased data structure implies a consistent, robust learning-augmented data structure via thresholding. In this field, each key $k_i$ is associated with a weight $w_i > 0$. 
We call a data structured \emph{biased} if it is a dynamic data structure with the following guarantee. After inserting it into the structure, retrieving key $k_i$ takes $O(\log W/w_i)$ time,
where $W = \sum_{i=1}^n w_i$. Several data structures, including the binary search tree~\cite{bent1985biased,atallah1994biased}, skip-list~\cite{bagchi2005biased}, 
treap~\cite{seidel1996randomized},
zip tree~\cite{tarjan2021zip,gila2023zip}, zip-zip tree~\cite{gila2023zip}, skip-list tree~\cite{Erickson,bagchi2005biased}, B-tree~\cite{feigenbaum1983two}, and energy-balanced tree~\cite{goodrich2000competitive}
can be implemented as biased data structures.\footnote{
Some of the cited structures hold the result in expectation.
}
See also, e.g.,~\cite{bose2014biased,bose2013history,goodrich2010priority,arya2007simple}.

\begin{lemma}\label{lem:biased-consistent}
Any biased data structure $\mathcal B$ is consistent in the learning-augmented setting.
\end{lemma}
\begin{proof}
In the learning-augmented setting, we set $w_i \leftarrow f_i$. Thus, $W = \sum_{i=1}^nf_i \leq 1$. Therefore, each key $k_i$ can be found in time $O(\log 1/f_i)$ by properties of $\mathcal B$. 
\end{proof}

\begin{theorem}
Any biased data structure can be made into a consistent, robust learning-augmented data structure via  thresholding
\end{theorem}
\begin{proof}
Follows from Lemma~\ref{lem:biased-consistent} and Theorem~\ref{thm:thresh-learning}.
\end{proof}

This result implies several new data structures that match state-of-the-art guarantees in the learning-augmented setting. 
In particular, our results also imply the 
first consistent and robust external-memory learning-augmented data structure by applying thresholding to~\cite{chen2025on,feigenbaum1983two}. Our scheme also satisfies an open question of~\cite{zeynali2024robust}. They asked whether one can maintain a consistent and robust data structure in which some keys have frequency estimates and others do not. We can easily handle this by assuming any key without an estimate has frequency 0. These keys will be thresholded and have $O(\log n)$ search time without affecting the search time of other keys. 

Lastly, any
of the above structures can made dynamic, 
i.e. allow for changes in $n$, either by using the amortized update scheme proposed by 
Zeynali {\it et al.} for \texttt{RobustSL}~\cite{zeynali2024robust} or, 
if the underlying data structure is history independent, 
using the history independent updating scheme we propose in Section~\ref{sec:history-ind-scheme}. 
In this case, we would set $f_i' \leftarrow \max\{f_i/2, 1/(2N)\}$ and achieve similar performance and robustness guarantees as~\cite{zeynali2024robust}.
Thus, we can also implement Theorem~\ref{thm:hist-ind} using the following history independent data structures: the zip-tree~\cite{tarjan2021zip}, zip-zip tree~\cite{gila2023zip}, 
skip-list tree~\cite{gila2023zip,Erickson},
treap~\cite{lin2022learning,chen2025on}, and B-treap~\cite{chen2025on,golovin2008uniquely}.

\section{Pairing: Obtaining Strong History Independence}

Our results in Sections~\ref{sec:history-ind-scheme} and~\ref{sec:thresholding} implied the first weakly history independent, dynamic, consistent, and robust learning-augmented data structures. Here, we develop an orthogonal technique we call \emph{pairing} that also transforms any consistent data structure into a robust one while additionally providing \emph{strong} history independence. 

Define a \emph{paired data structure} $\mathcal D_P$ as two copies of the data stored in two separate data structures $\mathcal D_C$ and $\mathcal D$.\footnote{
To save space, $\mathcal D_C$ and $\mathcal D$ may store pointers to shared data rather than maintaining actual copies.
}
$\mathcal D_C$ is any strongly history independent consistent learning-augmented data structure (e.g., the biased zip-zip tree~\cite{gila2023zip} via Lemma~\ref{lem:biased-consistent}) and $\mathcal D$ is a strongly history independent non-learned data structure with support for inexact queries and $O(\log n)$-retrieval time (e.g., a standard skip-list or zip-zip tree~\cite{gila2023zip}), where $n$ is the number of keys in the structure currently. 

We insert and delete elements from both structures in tandem. To search, we do the following. We begin a tentative search in $\mathcal D_C$ for $\gamma \log n$ steps for some constant $\gamma > 0$. If we find the element, we are done. If we have made $\gamma \log n$ steps and have not found the element, we terminate the search in $\mathcal D_C$ and search $\mathcal D$ for the element. For some query key $k_i$ with frequency $f_i$ that exists in $\mathcal D_P$, it immediately follows that this takes $O(\min\{\log 1/f_i, \log  n\})$ time, satisfying both consistency and robustness. 
Using standard techniques, we can implement predecessor and range queries in $O(\log n)$-time in $\mathcal D$, matching the results of~\cite{zeynali2024robust}. 

Unlike the robust learning-augmented data structures of~\cite{fu2024learning,zeynali2024robust} and our thresholding scheme from above, the underlying structure of $\mathcal D_P$ does not depend on the number of keys inserted, $n$ (apart from its size). 
Thus, we do not need an amortized updating scheme in which the structure periodically rebuilds itself to maintain its correctness.
As a result, the strong history independence of 
$\mathcal D_C$ and $\mathcal D$ 
is preserved in the dynamic setting and we immediately have the following.

\begin{theorem}
There exists a strongly history independent, dynamic, consistent, and robust learning-augmented data structure.
\end{theorem}
\begin{proof}
Follows from instantiating a paired data structure using a strongly history independent consistent data structure $\mathcal D_C$, e.g.~\cite{gila2023zip,tarjan2021zip,bagchi2005biased,chen2025on}, and a strongly history independent non-learned data structure $\mathcal D$ which supports inexact queries, e.g.~\cite{gila2023zip,tarjan2021zip,golovin2008uniquely,chen2025on}.
\end{proof}

Thus, we have achieved a stronger privacy guarantee at the cost of doubling space usage and increasing search times by a constant factor for less-frequent elements. 
In addition, unlike our thresholding scheme, pairing is unable to support keys that have not been assigned a frequency. 
In the following section, we collect experimental data to determine the impact of these tradeoffs in practice. 
\section{Experiments}

In this section, we discuss experimental results in the static setting, i.e., with fixed $n$.
We compare the biased zip-zip tree~\cite{gila2023zip} with our thresholding scheme against
other learning-augmented data structures, including \texttt{RobustSL}~\cite{zeynali2024robust}. 
In general, we find that our consistent, robust zip-zip tree has comparable or better performance to \texttt{RobustSL} while 
also being much easier to implement. Indeed, \texttt{RobustSL} requires a ground-up rewrite of skip-lists and
introduces several tunable parameters that must be optimized. In contrast, converting a 
zip-zip tree implementation to a biased zip-zip tree with thresholding requires changing fewer than 
five lines of code. Similarly, implementing a paired zip-zip tree is straightforward: one initializes two 
zip-zip trees, one standard and one biased, performs all updates in tandem, and applies the fall-back search procedure described above. 

We test the \texttt{RobustSL}~\cite{zeynali2024robust}, biased zip-zip tree~\cite{gila2023zip}, 
biased zip-zip tree with thresholding (threshold zip-zip tree), paired data structure with $\gamma = 1$ implemented via zip-zip trees (paired zip-zip tree), 
the learning-augmented treaps of Lin {\it et al.}~\cite{lin2022learning} (L-Treap) and Chen {\it et al.}~\cite{chen2025on} (C-Treap), 
as well as the (non-learned) AVL tree~\cite{AVL}. 
Following~\cite{zeynali2024robust}, we consider the Zipfian distribution~\cite{powers-1998-applications}, commonly used to model text frequencies, 
in which each key of rank $i$ (1 to $n$) is assigned frequency $1/i^\alpha$, for some parameter $\alpha \geq 1$. 
The keys are inserted from $1$ to $n$ into the respective structure with these frequencies, 
100000 queries are made, where each query of key $i$ occurs with frequency $1/i^\alpha$, 
and we count the average number of comparisons made over all queries. 

To test robustness,~\cite{zeynali2024robust} use the following ``noisy frequency'' scheme. 
They define a parameter $\delta\in [0,1]$ and 
adversarial rank $\hat i \leftarrow i \times (1 - \delta) + \delta \times (n - i + 1)$.
We insert the keys under these adversarial ranks (exactly reversed when $\delta = 1$) but perform queries using the original 
ranks. Thus, the learning-augmented structures are tested with adversarially bad frequency estimates. 
To further emphasize the influence of adversarially-chosen frequencies, we also define the inverse power distribution, in which each key of
rank $i$ is assigned frequency $1/\alpha^i$ for some $\alpha \geq 1$. 

In power law relations, such as the Zipfian distribution, the smallest frequencies are still inverse polynomial in $n$. 
Thus, any non-robust learning-augmented data structure
should still expect to have $O(\log n)$ retrieval times.
In contrast,
the smallest frequencies in the inverse power distribution are inverse exponential in $n$, which
may degrade to linear search times in non-robust structures under adversarially-chosen frequencies.
As a result, this distribution better demonstrates the power of robustness in the learning-augmented setting.

Our tests are as follows. 
The Zipf parameter test examines different values of $\alpha$ under perfect frequency estimates ($\delta = 0$) with a fixed $n = 2000$; 
the Noisy Zipfian test uses $\alpha = 2$ and $\delta = 0.9$; and the 
Inverse Power test uses $\alpha = 1.01$ (chosen for computational tractability) and $\delta = 0.9$.
Lastly, our size test measures \texttt{RobustSL}'s node count when its keys follow the Zipfian distribution with $\alpha = 2$. 
The other structures considered have exactly $n$ nodes, 
except the paired zip-zip tree, which has $2n$ nodes. 


\begin{figure*}
    \centering
    \includegraphics[width=\linewidth]{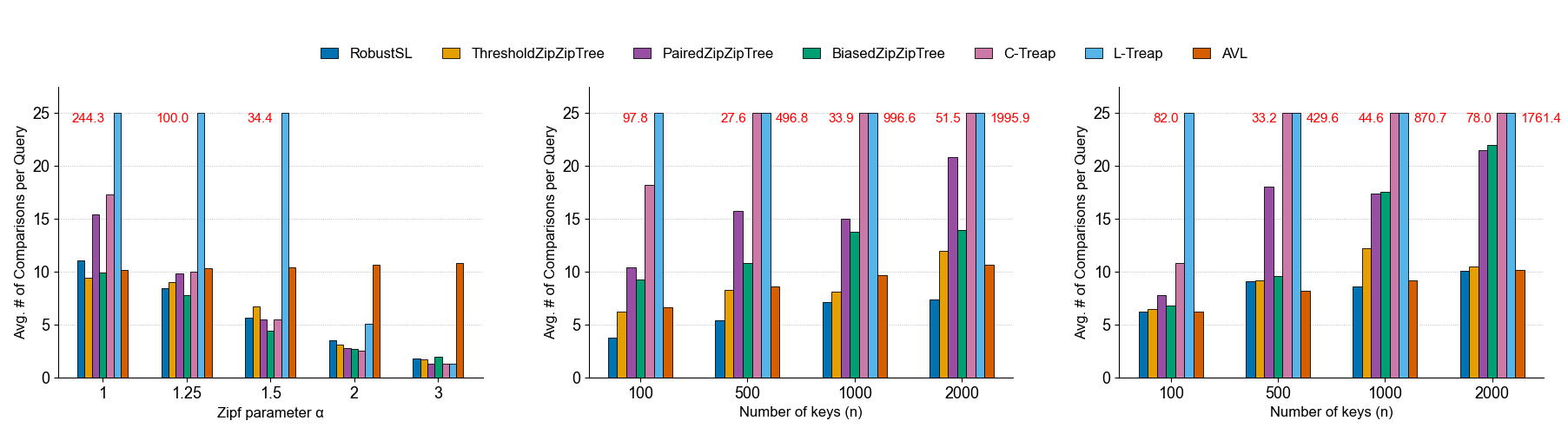}
    \caption{(a) Zipf Parameter test ($n=2000$, $\delta=0$), (b) Noisy Zipfian test ($\alpha = 2$, $\delta=0.9$), and (c) Inverse Power test ($\alpha =1.01$, $\delta=0.9$).
    Values overflowing 25 comparisons indicated with a number next to the bar.}
    \label{fig:zipf}
\end{figure*}

\begin{figure}
    \centering
    \includegraphics[width=\linewidth]{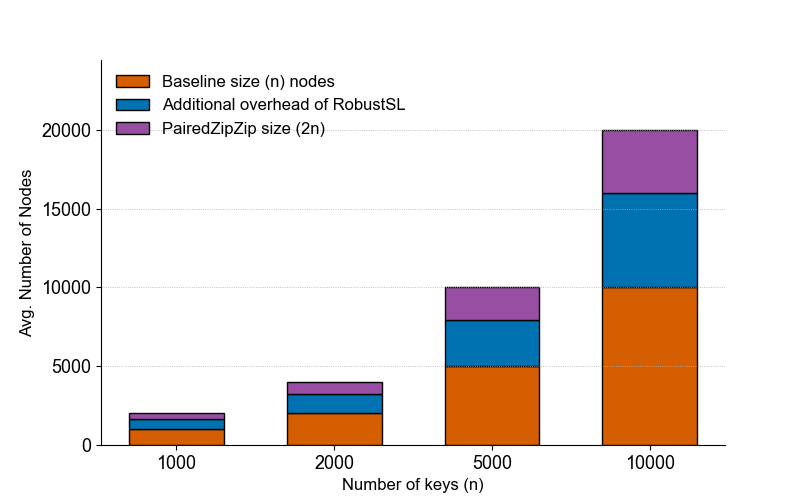}
    \caption{Size test (Zipfian, $\alpha=2$)}
    \label{fig:size}
\end{figure}


\subsection{Experimental Results}

We find that the threshold zip-zip tree achieves query performance within 2 comparisons of \texttt{RobustSL} on average while only using two-thirds of the space, demonstrating an advantageous balance between efficiency and memory cost. 
We also find that the paired zip-zip tree achieves query performance within at most $2 \times$ the comparisons of \texttt{RobustSL} and the threshold zip-zip tree on average, while using $\approx 25 \%$ more space than \texttt{RobustSL}.

\textbf{Zipfian Workloads.} 
In Figure~\ref{fig:zipf}a, we perform the Zipfian Parameter Test, which varies $\alpha$ under perfect predictions. 
As the skew increases, we can see that all learned structures are indeed consistent and take advantage of the frequency distribution.
The non-robust structures (biased zip-zip, L-Treap, and C-Treap) benefit the most from a higher skew. 
Nevertheless, \texttt{RobustSL} and the threshold zip-zip tree trail by just 2-3 comparisons on average and remain closely matched with each other. 
The paired zip-zip tree also performs as well as the others as $\alpha$ increases, though has poorer constant factors for $\alpha = 1$.

Next, we introduce noise into the predictions with $\alpha=2$ and $\delta=0.9$ (see Figure~\ref{fig:zipf}b). 
We observe that the
non-robust structures C-Treap and L-Treap degrade sharply, with average comparisons growing linearly with $n$.\footnote{
The L-Treap's generally poor performance is also due to the fact that it requires the input to be inserted in uniformly permuted order to remain balanced~\cite{lin2022learning,zeynali2024robust,chen2025on}.
}
In contrast, both \texttt{RobustSL} and threshold zip-zip tree stay below $7$ comparisons on average, with the threshold zip-zip tree trailing by at most 2-3 comparisons.
They also still take advantage of the frequency distribution despite the noise, consistently performing better than the AVL tree.
The paired zip-zip tree remains within 2$\times$ the cost of the threshold zip-zip tree throughout. 
Interestingly, the biased zip-zip tree is less affected than both treaps, though it still performs worse than the threshold zip-zip tree and \texttt{RobustSL}. 





\textbf{Inverse Power Workloads.} 
We designed the inverse power test to show that, under schemes in which frequencies grow exponentially small in $n$, 
adversarial predictions force linear performance in any consistent, non-robust learned data structure (see Figure~\ref{fig:zipf}c).
The data supports this hypothesis, as all three non-robust structures (biased zip-zip, L-Treap, and C-Treap) appear to scale linearly in $n$.
In contrast, threshold zip-zip tree and \texttt{RobustSL} match the roughly $\log n$ growth of the AVL tree.
Furthermore, despite performing worse than the biased zip-zip tree for small $n$, the paired zip-zip tree has a clear logarithmic growth of its query cost compared to the linear growth of the biased zip-zip tree. 
Thus, we have experimental evidence that our methods provide robustness even in an extreme adversarial setting.


\textbf{Space Overhead.} 
Another consideration in choosing an appropriate database structure is space usage.
All tree-based structures considered maintain exactly $n$ nodes, while \texttt{RobustSL} incurs significant overhead due to its skip-list layering structure. 
Accordingly,
we find that \texttt{RobustSL} uses roughly $1.5\times$ space,
which may be a limiting factor if applied to large-scale
datasets.
The paired zip-zip tree maintains two representations of the dataset, yielding a space cost of $2n$ (see Figure~\ref{fig:size}).


Our results suggest that the threshold zip-zip tree may offer a simpler and practical alternative to \texttt{RobustSL} for learning-augmented workloads, achieving comparable robustness and performance under synthetic workloads while significantly reducing space usage.
In addition, our results for the paired zip-zip tree show a clear tradeoff between security and efficiency. The paired zip-zip tree offers strong history independence in the dynamic setting, yet at the cost of a 2$\times$ increase in space and time usage compared to our threshold zip-zip tree. 

\section{Conclusion}
In this paper, we initiated a study of privacy and security in 
learning-augmented data structures by proposing the first 
dynamic, consistent, robust, and history independent learned data structures
using our new techniques, thresholding and pairing. 
Our techniques also 
more than triple the number of practical learning-augmented data structures that are consistent and robust. 

Future work could experimentally validate our new data structures in the dynamic setting or examine whether better choices for $\gamma$ exist in the paired zip-zip tree (e.g., $\gamma = \{3.82, 1.3863\}$ for height and expected depth of a standard zip-zip tree respectively~\cite{gila2023zip}). 
One could also consider how to narrow the efficiency gap between our strongly and weakly history independent robust learned data structures. 

We also observe that, 
in general, any biased data structure in which a bound $W^*$ on $W$ is known in advance can be 
implemented with thresholding such that each key with weight $w_i$ has depth $O(\min\{\log W^*/w_i, \log n\})$. 
We wonder if this could be useful elsewhere, e.g., in variants of Sleator and Tarjan's link-cut tree~\cite{sleator1981data} or to generalize biased data structures which require polynomially-bounded weights such as Goodrich and Strash's priority range tree~\cite{goodrich2010priority}.




\bibliographystyle{plain}
\bibliography{refs}

\begin{thebibliography}{10}

\bibitem{AVL}
G.~Adelson-Velskii and E.~Landis.
\newblock An algorithm for the organization of information, 1962.

\bibitem{arya2007simple}
Sunil Arya, Theocharis Malamatos, and David~M Mount.
\newblock A simple entropy-based algorithm for planar point location.
\newblock {\em ACM Transactions on Algorithms (TALG)}, 3(2):17--es, 2007.

\bibitem{atallah1994biased}
Mikhail~J Atallah, Michael~T Goodrich, and Kumar Ramaiyer.
\newblock Biased finger trees and three-dimensional layers of maxima: (preliminary version).
\newblock In {\em Proceedings of the tenth annual symposium on Computational geometry}, pages 150--159, 1994.

\bibitem{bagchi2005biased}
Amitabha Bagchi, Adam~L Buchsbaum, and Michael~T Goodrich.
\newblock Biased skip lists.
\newblock {\em Algorithmica}, 42(1):31--48, 2005.

\bibitem{benomar2024learning}
Ziyad Benomar and Christian Coester.
\newblock Learning-augmented priority queues.
\newblock {\em Advances in Neural Information Processing Systems}, 37:124163--124197, 2024.

\bibitem{benomar2025tradeoffs}
Ziyad Benomar and Vianney Perchet.
\newblock On tradeoffs in learning-augmented algorithms.
\newblock {\em arXiv preprint arXiv:2501.12770}, 2025.

\bibitem{bent1985biased}
Samuel~W Bent, Daniel~D Sleator, and Robert~E Tarjan.
\newblock Biased search trees.
\newblock {\em SIAM Journal on Computing}, 14(3):545--568, 1985.

\bibitem{bose2014biased}
Prosenjit Bose, Rolf Fagerberg, John Howat, and Pat Morin.
\newblock Biased predecessor search $\delta$.
\newblock {\em LATIN 2014: Theoretical Informatics LNCS 8392}, page 755, 2014.

\bibitem{bose2013history}
Prosenjit Bose, John Howat, and Pat Morin.
\newblock A history of distribution-sensitive data structures.
\newblock In {\em Space-Efficient Data Structures, Streams, and Algorithms: Papers in Honor of J. Ian Munro on the Occasion of His 66th Birthday}, pages 133--149. Springer, 2013.

\bibitem{chen2015hiflash}
Bo~Chen and Radu Sion.
\newblock Hiflash: A history independent flash device.
\newblock {\em arXiv preprint arXiv:1511.05180}, 2015.

\bibitem{chen2025on}
Jingbang Chen, Xinyuan Cao, Alicia Stepin, and Li~Chen.
\newblock On the power of learning-augmented search trees.
\newblock In {\em Forty-second International Conference on Machine Learning}, 2025.

\bibitem{ding2020alex}
Jialin Ding, Umar~Farooq Minhas, Jia Yu, Chi Wang, Jaeyoung Do, Yinan Li, Hantian Zhang, Badrish Chandramouli, Johannes Gehrke, Donald Kossmann, et~al.
\newblock Alex: an updatable adaptive learned index.
\newblock In {\em Proceedings of the 2020 ACM SIGMOD international conference on management of data}, pages 969--984, 2020.

\bibitem{Erickson}
Jeff Erickson.
\newblock Lecture notes on treaps., 2017.

\bibitem{feigenbaum1983two}
Joan Feigenbaum and Robert~E Tarjan.
\newblock Two new kinds of biased search trees.
\newblock {\em Bell System Technical Journal}, 62(10):3139--3158, 1983.

\bibitem{ferragina2020learned}
Paolo Ferragina and Giorgio Vinciguerra.
\newblock Learned data structures.
\newblock In {\em Recent Trends in Learning From Data: Tutorials from the INNS Big Data and Deep Learning Conference (INNSBDDL2019)}, pages 5--41. Springer, 2020.

\bibitem{fu2024learning}
Chunkai Fu, Brandon~G Nguyen, Jung~Hoon Seo, Ryan Zesch, and Samson Zhou.
\newblock Learning-augmented search data structures.
\newblock {\em arXiv preprint arXiv:2402.10457}, 2024.

\bibitem{gila2023zip}
Ofek Gila, Michael~T Goodrich, and Robert~E Tarjan.
\newblock Zip-zip trees: Making zip trees more balanced, biased, compact, or persistent.
\newblock In {\em Algorithms and Data Structures Symposium}, pages 474--492. Springer, 2023.

\bibitem{golovin2008uniquely}
Daniel Golovin.
\newblock {\em Uniquely represented data structures with applications to privacy}.
\newblock Carnegie Mellon University, 2008.

\bibitem{goodrich2000competitive}
Michael~T Goodrich.
\newblock Competitive tree-structured dictionaries.
\newblock In {\em Proceedings of the eleventh annual ACM-SIAM symposium on Discrete algorithms}, pages 494--495, 2000.

\bibitem{goodrich2016more}
Michael~T Goodrich, Evgenios~M Kornaropoulos, Michael Mitzenmacher, and Roberto Tamassia.
\newblock More practical and secure history-independent hash tables.
\newblock In {\em European symposium on Research in Computer Security}, pages 20--38. Springer, 2016.

\bibitem{goodrich2017auditable}
Michael~T. Goodrich, Evgenios~M. Kornaropoulos, Michael Mitzenmacher, and Roberto Tamassia.
\newblock Auditable data structures.
\newblock In {\em 2017 IEEE European Symposium on Security and Privacy (EuroS\&P)}, pages 285--300, 2017.

\bibitem{goodrich2010priority}
Michael~T Goodrich and Darren Strash.
\newblock Priority range trees.
\newblock In {\em International Symposium on Algorithms and Computation}, pages 97--108. Springer, 2010.

\bibitem{hartline2005characterizing}
Jason~D Hartline, Edwin~S Hong, Alexander~E Mohr, William~R Pentney, and Emily~C Rocke.
\newblock Characterizing history independent data structures.
\newblock {\em Algorithmica}, 42(1):57--74, 2005.

\bibitem{hsu2019learning}
Chen-Yu Hsu, Piotr Indyk, Dina Katabi, and Ali Vakilian.
\newblock Learning-based frequency estimation algorithms.
\newblock In {\em International Conference on Learning Representations}, 2019.

\bibitem{data-breach}
ITRC.
\newblock Itrc 2024 annual data breach report, Jan 2025.

\bibitem{kraska2018case}
Tim Kraska, Alex Beutel, Ed~H Chi, Jeffrey Dean, and Neoklis Polyzotis.
\newblock The case for learned index structures.
\newblock In {\em Proceedings of the 2018 international conference on management of data}, pages 489--504, 2018.

\bibitem{lin2022learning}
Honghao Lin, Tian Luo, and David Woodruff.
\newblock Learning augmented binary search trees.
\newblock In {\em International Conference on Machine Learning}, pages 13431--13440. PMLR, 2022.

\bibitem{mitzenmacher2018model}
Michael Mitzenmacher.
\newblock A model for learned bloom filters and optimizing by sandwiching.
\newblock {\em Advances in neural information processing systems}, 31, 2018.

\bibitem{mitzenmacher2019scheduling}
Michael Mitzenmacher.
\newblock Scheduling with predictions and the price of misprediction.
\newblock {\em arXiv preprint arXiv:1902.00732}, 2019.

\bibitem{mitzenmacher2022algorithms}
Michael Mitzenmacher and Sergei Vassilvitskii.
\newblock Algorithms with predictions.
\newblock {\em Communications of the ACM}, 65(7):33--35, 2022.

\bibitem{molnar2006tamper}
David Molnar, Tadayoshi Kohno, Naveen Sastry, and David Wagner.
\newblock Tamper-evident, history-independent, subliminal-free data structures on prom storage-or-how to store ballots on a voting machine.
\newblock In {\em 2006 IEEE Symposium on Security and Privacy (S\&P'06)}, pages 6--pp. IEEE, 2006.

\bibitem{naor2001anti}
Moni Naor and Vanessa Teague.
\newblock Anti-persistence: History independent data structures.
\newblock In {\em Proceedings of the thirty-third annual ACM symposium on Theory of computing}, pages 492--501, 2001.

\bibitem{nathan2020learning}
Vikram Nathan, Jialin Ding, Mohammad Alizadeh, and Tim Kraska.
\newblock Learning multi-dimensional indexes.
\newblock In {\em Proceedings of the 2020 ACM SIGMOD international conference on management of data}, pages 985--1000, 2020.

\bibitem{polak2024learning}
Adam Polak and Maksym Zub.
\newblock Learning-augmented maximum flow.
\newblock {\em Information Processing Letters}, 186:106487, 2024.

\bibitem{powers-1998-applications}
David M.~W. Powers.
\newblock Applications and explanations of {Z}ipf{'}s law.
\newblock In {\em New Methods in Language Processing and Computational Natural Language Learning}, 1998.

\bibitem{purohit2018improving}
Manish Purohit, Zoya Svitkina, and Ravi Kumar.
\newblock Improving online algorithms via ml predictions.
\newblock {\em Advances in Neural Information Processing Systems}, 31, 2018.

\bibitem{qin2025piecewise}
Jiayong Qin, Xianyu Zhu, Qiyu Liu, Guangyi Zhang, Zhigang Cai, Jianwei Liao, Sha Hu, Jingshu Peng, Yingxia Shao, and Lei Chen.
\newblock Piecewise linear approximation in learned index structures: Theoretical and empirical analysis.
\newblock {\em arXiv preprint arXiv:2506.20139}, 2025.

\bibitem{sato2023fast}
Atsuki Sato and Yusuke Matsui.
\newblock Fast partitioned learned bloom filter.
\newblock {\em Advances in Neural Information Processing Systems}, 36:39119--39146, 2023.

\bibitem{seidel1996randomized}
Raimund Seidel and Cecilia~R Aragon.
\newblock Randomized search trees.
\newblock {\em Algorithmica}, 16(4):464--497, 1996.

\bibitem{sleator1981data}
Daniel~D Sleator and Robert~Endre Tarjan.
\newblock A data structure for dynamic trees.
\newblock In {\em Proceedings of the thirteenth annual ACM symposium on Theory of computing}, pages 114--122, 1981.

\bibitem{sleator1985self}
Daniel~Dominic Sleator and Robert~Endre Tarjan.
\newblock Self-adjusting binary search trees.
\newblock {\em Journal of the ACM (JACM)}, 32(3):652--686, 1985.

\bibitem{tarjan2021zip}
Robert~E Tarjan, Caleb Levy, and Stephen Timmel.
\newblock Zip trees.
\newblock {\em ACM Transactions on Algorithms (TALG)}, 17(4):1--12, 2021.

\bibitem{zeynali2024robust}
Ali Zeynali, Shahin Kamali, and Mohammad Hajiesmaili.
\newblock Robust learning-augmented dictionaries.
\newblock In {\em International Conference on Machine Learning}, pages 58470--58483. PMLR, 2024.

\end{thebibliography}

\appendix
\section{History Independence}\label{sec:hist-ind-ex}
A data structure is history independent~\cite{hartline2005characterizing,naor2001anti} if its history of operations, 
such as insertions, deletions, or queries, cannot be inferred from
its internal memory representation beyond what is implied
by its current data contents, known as its \emph{state}. 
There are two types of history independence,
\emph{strong} and \emph{weak}. 
Intuitively, one can think of strong history independence as preventing data leakage across multiple data breaches, whereas weak history independence guarantees protection for only a single breach.

Formally, let us consider 
two states $A$ and $B$, and two sets of operations $X$ and $Y$. Let 
$A$ be the initial state and $B$ be the final state of our data
structure.

\begin{definition}[Strong History Independence]
    A data structure is \textit{strongly history independent} if, for 
    any two sets of operations $X$ and $Y$ that result in the same 
    logical state $B$, the resulting internal memory representation induced by each set of operations is indistinguishable.  
\end{definition}


\begin{definition} [Weak History Independence]
    A data structure is \textit{weakly history independent} if this 
    indistinguishable property holds only when $A$ is the initial (e.g., empty) state. 
\end{definition}


\end{document}